\documentclass[10pt]{article}

    \usepackage{epsfig}
    \usepackage{amsmath,amsthm}
    \usepackage{amstext}
    \usepackage{amsfonts}
    \usepackage{amssymb}

    \usepackage{eucal}
    \usepackage{graphicx}

\newcommand{\ls}[1]
   {\dimen0=\fontdimen6\the\font \lineskip=#1\dimen0
\advance\lineskip.5\fontdimen5\the\font \advance\lineskip-\dimen0
\lineskiplimit=.9\lineskip \baselineskip=\lineskip
\advance\baselineskip\dimen0 \normallineskip\lineskip
\normallineskiplimit\lineskiplimit \normalbaselineskip\baselineskip
\ignorespaces }
\ls{1.0}

\newtheorem{definition}{Definition}[section]
\newtheorem{lemma}{Lemma}[section]

\newtheorem{theorem}{Theorem}[section]
\newtheorem{example}{Example}[section]

\begin{document}
\baselineskip 0.1in

\title{Monotonic Preference Aggregation Mechanisms for Buying a Shareable Resource \thanks{This material is based upon
work supported in part by the U.S. Air Force Office of Scientific Research (AFOSR) under grant number MURI FA9550-10-1-0573.} }
\author{Vijay Kamble and Jean Walrand\\
Dept. of Electrical Engineering and Computer Sciences,\\
University of California, Berkeley\\
{vjk,wlr}@eecs.berkeley.edu}
\maketitle

\ls{0.95}

\date{}

\maketitle

\begin{abstract}
Situations where a group of agents come together to jointly buy a resource that they individually cannot afford to buy are commonly observed in markets. 
For example in the US market for radio spectrum, a recent proposal invited small firms who would benefit from gaining additional access to spectrum to jointly submit bids for blocks of spectrum with the idea that its utilization could be shared. In such a scenario, the problem is to design a mechanism that truthfully elicits and aggregates the privately held preferences of these agents, and enables them to act as a single decision-making body in order to participate in the market. In this paper, we design a class of mechanisms called \emph{monotonic aggregation mechanisms} that achieves this under a specific setting. We assume that the resource is being sold in a sealed-bid second-price auction that solicits bids for the entire resource. Our mechanism truthfully elicits utility functions from the buyers, prescribes a joint bid, and prescribes a division of the payment and the resource in the event that they win the resource in the auction. This mechanism further satisfies a popular notion of collusion-resistance known as coalition-strategyproofness. We give two explicit examples of this generic class for the case where the utility functions of the buyers are non-decreasing and concave.

\end{abstract}

\ls{0.95}

\date{}

\maketitle
\section{Introduction}
Examples of resources that are produced and sold as discrete units on the supply side, but which could be shared or be divided amongst buyers on the demand side are ubiquitous. For example there is a single fare to be paid for a taxi-ride no matter how many people take it, but multiple people can potentially share this fare; or an entire house can be rented for a fixed monthly payment that can be shared between roommates. There are other examples like laboratories jointly purchasing an expensive scientific equipment for experimentation or firms buying some advanced computing resource to be shared amongst themselves.

In many cases, this is because the resource by its nature itself can only produced in discrete quantities but can be shared, or it could also be because, although the resource is divisible, it is easier and perhaps more cost-efficient to sell it in discrete units, where the size of these units is decided on the basis of typical nature of demand. For example, consider the market for radio spectrum, where the spectrum is sold in the form of blocks of contiguous frequencies. In the US market, the majority of the current demand for spectrum is for the purpose of mobile broadband communications, and the blocks are designed with these requirements in mind (e.g., see the upcoming incentive auctions for repurposing the spectrum used by TV broadcasters for mobile communications (LTE) \cite{incentive_auction_ro}). The players in this market are the large cellular service providers like AT\&T, Verizon, T-mobile etc. who compete for several blocks of spectrum in different geographical regions nationwide. But there is another pool of small interested parties who advocate the allocation of these blocks for free unlicensed use (e.g. WiFi) as opposed to selling them for exclusive licensed use \footnote{It has been argued (see for example \cite{milgrom2011case}) that an open-access unlicensed spectrum can act as an enabler for technological innovations that would increase social welfare and future tax revenue for the government, while at the same time mitigating the inefficiencies of exclusive licensed use. But opponents argue that an unlicensed spectrum may cause the government to lose out on the substantial revenue that it generates by sale to licensed users, which could have also been devoted to social benefit.}.
In a proposal in \cite{bykowsky2008market} that was further analyzed in \cite{Bykowsky10}, the authors suggest that groups consisting of smaller content or service providers who benefit from additional access to spectrum could jointly submit bids for a shared license, which will compete with bids for exclusive licenses from the bigger firms. 

Participating in such a market as a group with the intention of sharing the bought resource entails the collective decision-making problem of performing aggregate purchasing decisions when the utilities of the different buyers in the group are privately held. For example, two firms who intend to jointly bid in the market for radio spectrum have to decide how many blocks to bid for, how much should the bid be and how to divide the spectrum and the payment in the event that they win any spectrum. Naturally, all of these decisions have to go hand-in-hand and any scheme that helps them take these joint decisions fairly has to rely on its ability to elicit the privately known utilities of these agents. The goal of this paper is to design preference aggregation mechanisms for making such joint purchasing decisions. 

We focus on the case where there is a single discrete resource offered for sale in a sealed-bid second price auction that accepts bids for purchasing the \emph{entire} resource. Note that this includes the setting where there is no competition and the seller offers the resource for a fixed price. Further, we assume that the resource is divisible or shareable and the buyer's preferences are captured by a utility function that assigns real non-negative values to different proportions in which resource can be utilized. We will focus on the case where these utility functions are concave, non-negative and non-decreasing. A group of buyers then needs to decide what bid is to be submitted to the auction, and how the resource and payment should be divided in the event that they win the resource. 

We present a class of mechanisms called \emph{monotonic aggregation mechanisms} that enable a group to take these collective decisions. This mechanism truthfully elicits utility functions from the buyers, prescribes a joint bid, and prescribes a division of the payment and the resource in the event that they win in the auction. This mechanism moreover satisfies a popular notion of collusion-resistance known as coalition-strategyproofness (also called group-strategyproofness in literature). A mechanism is coalition-strategyproof if no coalition of buyers can find a deviation from truthfulness such that no buyer in the coalition is worse off and at least one buyer is strictly better off, irrespective of the reports of the buyers not in the coalition. 

The most important component of this mechanism is the specification of the resource and payment shares for different contingencies that may arise in the auction. These shares need to satisfy a `monotonicity' condition in order for the incentive properties of the mechanism to hold (hence the name). This condition is implicitly defined with respect to the class of functions that the buyers' utility functions belong to. We give two explicit instances of this mechanism that satisfy this condition for the case where these utility functions are concave, non-negative and non-decreasing. 
\section{Model: Buying a divisible resource}
A set (or group) $L$ of $n$ agents would like to collectively buy a single resource that they intend to share. We assume that the resource is being sold in a sealed-bid second-price auction with some reserve price $\omega \geq 0$, in which the resource is allotted to the highest bidder as long as his bid is higher than $\omega$, at the price equal to the maximum of the second highest bid and $\omega$ \footnote{Our results in fact hold for any deterministic dominant strategy truthful auction.}. 


The resource is assumed to be divisible and each agent $i$ in the group has a utility $U_i(x_i)$, expressed in monetary value, for a fraction $x_i$ of the resource. We assume that $U_i(x_i)$ is the maximum payment that a buyer is willing to make for the fraction $x_i$ of the resource. This utility function $U_i:[0,1]\rightarrow \mathbb{R}_{+}$ is known only to buyer $i$ and the functions $\{U_i: i\in L\}$ are assumed to belong the class $\mathcal{C}$ defined as follows.

\begin{definition}
$\mathcal{C}$ is the class of concave and non-decreasing functions $U:[0,1]\rightarrow \mathbb{R}_{+}$, such that $U(0)=0$.
\end{definition}

In order to participate in the auction, the group has to submit a single bid for the resource, and make the required payment computed by the second-price auction rule in the case that they win. Our goal is to design a mechanism that accomplishes the following two tasks:
\begin{enumerate}
\item Elicit individual utility functions from the agents and then output a group bid to enter into the external auction.
\item Prescribe a division of the resource and that of the payment needed to be paid in the external auction amongst the buyers, in the event that they win the resource.
\end{enumerate}
Such a mechanism will be called an \emph{aggregation mechanism} since it aggregates the preferences of all the agents in a group to result in a single decision-making entity. 


\section{Monotonic Aggregation Mechanisms: a structural overview through examples}
We first illustrate the structure of our mechanism with the help of a couple of examples. 

\begin{example}{\bf (Buying a resource for a fixed price)} Consider a resource $A$ that is being sold for a fixed price $p^*$. Suppose the group $L$ that intends to buy the resource consists of three buyers $1$, $2$ and $3$ with privately known utility functions $U_1(x)=x$, $U_2(x)=\sqrt{x}$ and $U_3(x)=\ln(1+x)$. First, for each subset $A\subseteq L$, the mechanism fixes two vectors corresponding to resource shares and payment shares respectively: $(x_1(A),\cdots,x_n(A))$ and $(y_1(A),\cdots,y_n(A))$ such that $\sum_{i\in A} x_i(A)=\sum_{i\in A} y_i(A)=1$. These shares satisfy certain additional conditions that will be defined later. We stress here that these shares are chosen by the mechanism before it solicits utility reports from the buyers.  For now, assume that mechanism chooses, $x_i(A)=y_i(A)=\frac{1}{|A|}$ if $i\in A$ and $0$ otherwise. This schedule of shares for the different subsets is announced to the buyers. The mechanism then solicits the utility functions of the buyers within the class $\mathcal{C}$ (i.e. the message space of possible reports is the class $\mathcal{C}$). Let these reported utility functions be $G_1$, $G_2$ and $G_3$. Once these reports are collected, the mechanism starts by first considering the entire set of buyers $S_1=L$.  Using the reported utility functions, it determines whether for each buyer $i\in S_1$, 
$$p^*y_i(S_1)\leq G_i(x_i(S_1)),$$
i.e. if $\frac{p^*}{|S_1|}\leq G_i(\frac{1}{|S_1|})$. If this condition is satisfied for all the buyers in  $S_1$ then it means that according to the reported utility functions, each buyer $i$ can afford to pay a fraction $y_i(S_1)$ of the price $p^*$ for a fraction $x_i(S_1)$ of the resource. The mechanism then buys the resource and both the resource and the price are divided according amongst the buyers according to the corresponding shares. If the condition is not satisfied for a set of buyers $T_1\subseteq S_1$, then they are removed to result in a smaller set of buyers $S_2=S_1\setminus T_1$, and the procedure is repeated until the mechanism either finds a subset of buyers that can together afford to buy the resource for the corresponding shares, or all the buyers are removed, in which case the resource is not bought. To illustrate this procedure, assume that the buyers are truthful in their reports, i.e. $G_i=U_i$ and that the price of the resource is $p^*=0.9$.
Then for the set $S_1=L=\{1,2,3\}$, mechanism checks if
\begin{eqnarray*}
(\frac{0.9}{3},\frac{0.9}{3},\frac{0.3}{3})=(0.3,0.3,0.3)\overset{?}{\preceq} (1/3,\, \frac{1}{\sqrt{3}},\, \ln(1+\frac{1}{3})),
\end{eqnarray*}
which does not hold since although $\frac{1}{3}>0.3$ and $ \frac{1}{\sqrt{3}}\sim 0.57>0.3$, $\ln(4/3)\sim 0.27 <0.3$.  Thus the mechanism removes buyer $3$ and considers remaining set of buyers $S_2=\{1,2\}$. It now checks whether
\begin{eqnarray*}
(\frac{0.9}{2},\frac{0.9}{2})=(0.45,0.45)\overset{?}{\preceq} (1/2,\, \frac{1}{\sqrt{2}}),
\end{eqnarray*}
which holds. Hence the resource is bought and both the resource and the price are split equally between buyers $1$ and $2$.
\end{example}

\begin{example} {\bf (Participating in an auction)} Consider now the same three buyers with utility functions $U_1(x)=x$, $U_2(x)=\sqrt{x}$ and $U_3(x)=\ln(1+x)$, but now assume that the resource $A$ is being sold in a sealed bid second-price auction. In this case also, for each subset $A\subseteq L$, the mechanism fixes two vectors, $(x_1(A),\cdots,x_n(A))$ and $(y_1(A),\cdots,y_n(A))$, corresponding to resource shares and payment shares respectively. Let us again assume that mechanism chooses, $x_i(A)=y_i(A)=\frac{1}{|A|}$ if $i\in A$ and $0$ otherwise, and this schedule of shares is announced to the buyers. The mechanism then solicits the utility functions of the buyers from within the class $\mathcal{C}$. Let $G_1$, $G_2$ and $G_3$ be these reported functions. The mechanism then starts by first considering the entire set of buyers $S_1=L$.  Using the reported utility functions, it computes the maximum cumulative payment $k$ such that for each buyer $i$ in $S_1$, $ky_i(S_1)\leq G_i(x_i(S_1))$. Denote this maximum value of $k$ as $\beta_1$. $\beta_1$ is thus the maximum total payment that the set of buyers $S_1$ can make, such that each buyer in $S_1$ can afford to buy his share of the resource in that set for the corresponding share of the total payment (i.e. his payment is less than his utility for his share). To illustrate this operation, suppose that the three buyers had reported their utility functions truthfully. Then for the set $S_1=L=\{1,2,3\}$, mechanism computes
\begin{eqnarray*}
\beta_1&=&\max\{k\geq 0: (\frac{k}{3},\frac{k}{3},\frac{k}{3})\preceq (1/3,\, \frac{1}{\sqrt{3}},\, \ln(1+\frac{1}{3}))\}\\
&=&\min\{1,\,\sqrt{3},\, 3\ln(\frac{4}{3})\}\\
&=&3\ln(\frac{4}{3})\approx 0.86.
\end{eqnarray*}
Now let $T_1$ be the subset of buyers whose utility for their share exactly equals their share of the maximum payment $\beta_1$, i.e. the bottleneck buyers. Then the buyers in $T_1$ are removed from the set $S_1$ to result in the smaller set $S_2$. Thus in this case, buyer $3$ is removed from the group to form $S_2=\{1,2\}$. The mechanism continues with the remaining set $S_2$, finds the corresponding payment $\beta_2$, and continues so on to find the rest of the vector $\bar{\beta}$ in a similar way until no buyer remains. Thus we have

\begin{eqnarray*}
\beta_2&=&\max\{k\geq 0: (\frac{k}{2},\frac{k}{2})\preceq (1/2,\, \frac{1}{\sqrt{2}})\}\\
&=&\min\{1,\,\sqrt{2}\}=1.
\end{eqnarray*}
Here $T_2=\{1\}$ and hence buyer $1$ is removed from $S_2$ to result in $S_3=\{2\}$, and we finally have
\begin{eqnarray*}
\beta_3&=&\max\{k\geq 0: k\leq 1\}=1.
\end{eqnarray*}
Hence the vector $\bar{\beta}=(0.86,1,1)$. The largest value in this vector is submitted to the auction. If a payment $p^*$ is to be made in the auction to win the resource, the mechanism looks for the largest subset $S_i$ (i.e. the one with the smallest index $i$) that can afford to pay the price, and both the price and the resource is divided according to the corresponding shares in that subset. Thus in this example $\beta^*=1$ is submitted as a bid in the auction. Suppose that there is a single other competing buyer in the auction and suppose that his bid is $0.6$. Thus the minimum payment required to win the auction for the group is $0.6$. Now $r=\min\{i:\beta_i\geq 0.6\}=1$. Thus the entire group, i.e. $S_1=L$ is allotted equal shares of the resource and each buyer pays $0.2$ $(\frac{0.6}{3})$ to the seller. Suppose instead that the other buyer in the auction submitted a bid of $0.9$. Then in that case $r=\min\{i:\beta_i\geq 0.9\}=2$. Thus the group $S_2=\{1,2\}$ is allotted equal shares of the resource, while buyer $3$ does not get any share of the resource. Both the winning buyers pay $0.45$ to the seller. Thus in short, the resource is shared between the largest subset of buyers who can jointly afford to pay the price, divided according to the prescribed shares for that subset.
\end{example}

These two examples are an instance of a class of mechanisms that are obtained by varying the resource and payment shares for the different subsets of buyers. In our main result in this paper, we show that if these shares are chosen in a way that they satisfy a certain monotonicity property in relation to the class of functions $\mathcal{C}$, then the mechanism is truthful and moreover, it is coalition-strategyproof. 
\begin{definition} ({\bf Coalition-strategyproofness}) An aggregation mechanism is coalition-strategyproof if for any coalition of buyers $S\subseteq L$, fixing any feasible utility function reports of all buyers not in $S$, for every feasible deviation of the buyers in $S$ from truthful reporting, either all the buyers are indifferent between the original outcome and the new resulting outcome or at least one buyer is strictly worse off.
\end{definition}
We then give explicit characterizations of sharing schedules that satisfy this property.

\section{Related work}
Our model is related to the problem of sharing the cost of a jointly utilized resource, which has a rich history in the economics literature. The typical model adapted to our setting is as follows (see \cite{billera1982allocation, mirman1982demand}. Given $n$ agents with demands $q_i$, the total cost for utilizing the resource is a function $C(q_1,\cdots,q_n)$. The problem is to define a division of this cost amongst the agents. 

There are two angles from which this problem has traditionally been approached. One is the normative approach which looks at the problem from the point of view of fairness, where the goal is to characterize the set of sharing rules that satisfy certain desirable axiomatic properties. The seminal work in this setting is the Shapley value sharing rule \cite{shapley1952value}, which considers the specific case of binary demand i.e. $q_i$ can either be $0$ or $1$, which can be interpreted as an agent participating or not. For a fixed order of the agents, each agent has a marginal cost for his participation in that order. Shapley value assigns to an agent the average marginal cost over all the possible $n!$ orders. This can be interpreted as each agent being charged the uniform average of the line integral of his marginal cost of participation along paths from $(0,0,\cdots,0)$ to $(1,1,\cdots,1)$ that traverse along the edges of the hypercube. This rule is extended to the case of variable demands by the Shapley-Shubik rule \cite{shubik1962incentives}: each agent is charged the uniform average of the line integral of his marginal cost of consumption along paths from $(0,0,\cdots,0)$ to $(q_1,q_2,\cdots,q_3)$ that traverse along the edges of the hypercuboid. The third popular rule is the Aumann-Shapley rule \cite{billera1982allocation, aumann1974values}, that charges each agent the integral of his marginal costs along the diagonal path from $(0,0,\cdots,0)$ to $(q_1,q_2,\cdots,q_3)$. Finally, the serial sharing rule \cite{moulin1992serial, friedman1999three} arranges the demands in an increasing order so that $q_1\leq q_2\leq\cdots\leq q_n$, and charges each agent the integral of his marginal costs along the path made up of line segments connecting $(0,0,\cdots,0)$ to $(q_1,q_1,\cdots,q_1)$ to $(q_1,q_2,\cdots,q_2)$ and so on to $(q_1,q_2,\cdots,q_n)$. Several axiomatizations have been proposed that characterize these `additive' rules, i.e. rules that act as linear operators on the cost structure, see \cite{sprumont2010axiomatization, moulin1995additive, friedman1999three}. Many non-additive rules have also been considered and axiomatized, see e.g. \cite{sprumont1998ordinal, moulin2007fair, koster2007moulin}.

The other approach, which is the one taken in this paper, brings incentives into focus. Here agents are assumed to be endowed with preferences over their consumption of the resource and the cost that they have to pay. Given a cost function and a cost-sharing rule, the agents play the `demand-game' where each player reports a demand and is charged according to the sharing rule. The problem is to characterize cost-sharing rules result in `demand games' that satisfy nice equilibrium properties. The seminal work on this approach \cite{moulin1999incremental} has focused on \emph{incremental cost sharing methods}, which for a reported vector of demands $(q_1,q_2,\cdots,q_n)$ charge each agent the integral of his marginal costs along some path from $(0,0,\cdots,0)$ to $(q_1,q_2,\cdots,q_n)$. As discussed earlier, examples of such rules are the serial cost sharing rule and the Aumann-Shapley rule. The incentive properties of these mechanisms crucially depend on the shape of the cost function. For the case of integral demands ($q_i$s are integers) and convex preferences, it has been shown that if the cost function has increasing marginal returns and supermodular ($\partial_{ii}C>0$ and $\partial_{ij}C>0$) then the incremental cost sharing mechanisms result in a demand game with a unique strong Nash-equilibrium welfare\footnote{A strong equilibrium is defined as a strategic profile for which no subset of players has a joint deviation that strictly benefits all of them, while all other players are expected to maintain their equilibrium strategies.}  (i.e. the strong equilibrium is unique or if there are multiple such equilibria, then they are welfare equivalent). In the case where the cost function is submodular with decreasing marginal returns ($\partial_{ii}C<0$ and $\partial_{ij}C<0$), only those incremental rules that integrate along paths from $(0,0,\cdots,0)$ to $(q_1,q_2,\cdots,q_n)$ along the edges of the hypercuboid, called the \emph{sequential standalone mechanisms} result in a demand game with this nice property. This result for submodular cost functions bears an important qualification: if the demands are binary, then any cost sharing scheme that satisfies a cross-monotonicity property results in a demand game with these desirable properties (see also \cite{moulinshenker99} for an independent study of this case). 

One can also look at these results from an implementation theory perspective: a social choice function maps the reported preferences of agents to an assignment of consumption levels for the agents and a division of the corresponding total cost. The goal is to characterize social choice functions that result in truthful reporting of preferences a `nice' equilibrium. For the two polar classes of cost functions, the respective incremental sharing mechanisms (or the cross-monotonic sharing mechanisms for the case of binary demands and submodular costs) result in a social choice function that is truthful and moreover coalition-strategyproof.  Of particular interest in literature is the case of additive cost functions, where $C(q_1,\cdots,q_n)=\overline{C}(q_1+q_2+\cdots,+q_n)$, and our focus in this paper is on cost functions of this form. In this case the supermodularity or submodularity conditions translate to the convexity and concavity respectively of the function $\overline{C}$.

The general problem of interest in our paper is designing sharing rules in the case where the cost is assigned to integral quantities of the resource, while the demand is continuous, resulting in a cost function of the form:
\begin{equation}
C(q_1,q_2,\cdots,q_n)=\left\{ \begin{array}{ll}
         \sum_{i=1}^{\lceil q_1+\cdots,+q_n\rceil} a_i & \mbox{if $q_1+\cdots,+q_n >0$};\\
         0& \mbox{otherwise}.\end{array} \right. 
\end{equation}
where $\lceil q\rceil$ is the smallest integer $i$ such that $i\geq q$ and $(a_1,a_2,\cdots)$ is the sequence of non-negative marginal cost for consuming each additional quantity of the resource. We look at the problem from the point of view of incentives and we take the implementation theoretic approach of truthfully eliciting preferences of the agents by proposing incentive-compatible resource and cost allocations, rather than designing a demand game by choosing a cost  sharing rule. If the demands are restricted to be integral and if the sequence $(a_1,a_2,\cdots)$ is either non-increasing or non-decreasing, then the incremental cost sharing schemes discussed earlier result in coalition-strategyproof social choice functions. But this is not true if the demands are continuous. Our mechanism is a step towards bridging this gap for the specific case of unit supply, which can be encoded in a cost function of the form: 

\begin{equation}
C(q_1,q_2,\cdots,q_n) = \left\{ \begin{array}{ll}
	0 & \mbox{if $q_1+\cdots,+q_n = 0$};\\
         p^* & \mbox{if $0<q_1+\cdots,+q_n \leq 1$};\\
        \infty & \mbox{if $q_1+\cdots,+q_n>1$}.\end{array} \right. 
\end{equation}

Finally, the general setting where the group of buyers jointly bid for the resource in an external auction, i.e. where the price $p^*$ is extraneously defined has not been considered before.

\section{Definition and main results}
In this section we first give a formal definition of the mechanism and then prove some of its properties. Let $\overline{\mathcal{C}}$ be a subset of the class of utility functions $U:[0,1]\rightarrow R_{+}$ that satisfy $U(0)=0$.\\
\hrule
\noindent {\bf Monotonic aggregation mechanism for the class of utility functions $\overline{\mathcal{C}}:$}\\ 
For each subset $A\subseteq L$, fix two tuples of $n$ non-negative numbers $(x_1(A),\cdots,x_n(A))$ and $(y_1(A),\cdots,y_n(A))$, corresponding to the resource shares and the payment shares respectively, such that the following three conditions are satisfied:
\begin{enumerate}
\item $\sum_{i=1}^n x_i(A)=1$ and $x_i(A)>0$ only if $i\in A$. 
\item $\sum_{i=1}^n y_i(A)=1$ and $y_i(A)>0$ only if $i\in A$. 
\item (monotonicity) For any $C>0$, for any two subsets $A$ and $B$ such that $A\subseteq B$, and for any $i\in A$, if $$U_i(x_i(B))< Cy_i(B)$$
then 
$$U_i(x_i(A))< Cy_i(A)$$ for every $U_i \in \overline{\mathcal{C}}$.
\end{enumerate}
The mechanism solicits utility function reports $G_i\in \overline{\mathcal{C}}$ from all the agents and computes a vector of values $$\bar{\beta}=(\beta_1,\beta_2,\cdots,\beta_m)$$ corresponding to diminishing subsets of agents $S_1\supset S_2\supset \cdots \supset S_m$ as follows. 
\begin{itemize}
\item Let $S_1=L$. For each subset $S_j$, define
\begin{eqnarray}\label{psbam}
\beta_j&=&\max\{k\geq 0\,: (ky_1(S_j),ky_2(S_j),\cdots,ky_n(S_j))\nonumber\\
&&\preceq \big(G_1(x_1(S_j)),\cdots,G_n(x_n(S_j))\big)\}
\end{eqnarray}
where $\preceq$ denotes a component-wise $\leq$ inequality.
\item Let $T_j$ be the set of agents with positive payment shares who force the inequality in the definition above, i.e. all agents $i\in S_j$ such that $y_i(S_j)>0$ and $\beta_jy_i(S_j)=G_i(x_i(S_j))$. Then $S_{j+1}=S_j\setminus T_j$ and $m$ is the smallest integer such that $S_{m+1}=\phi$.
\item Let $\beta^*=\max\{\beta_1,\cdots,\beta_m\}$. The mechanism submits the bid $\beta^*$ to the auction. 
\item Suppose the group wins the auction and has to make a payment $p^*\leq \beta^*$ . Let $r=\min\{i:\beta_i\geq p^*\}$. Then each agent $i$ pays $y_i(S_r)p^*$ and gets a fraction $x_i(S_r)$.
\end{itemize}
\hrule
\bigskip
$S_r$ will be called the winning set of buyers. The key requirement of the mechanism is that the sharing schedule should satisfy the implicitly defined monotonicity condition with respect to the class $\overline{\mathcal{C}}$. It says that if a buyer in a set cannot afford to pay his share of some price $C$ for his share of the resource in that set, then she should not be able to do so in any subset of that set, so long as his utility function is in $\overline{\mathcal{C}}$. Another way to state this requirement is to have 
$$\frac{U_i(x_i(B))}{y_i(B)}  \geq \frac{U_i(x_i(A))}{y_i(A)}$$
if $A \subset B$ for any $i\in A$ and for all $U\in \overline{\mathcal{C}}$. This means that the average utility per unit payment decreases for a buyer as the size of the subset decreases. In the next section we will give explicit sharing schedules that satisfy the monotonicity condition with respect to the class $\mathcal{C}$.

Also note that since the utility function report of a buyer is only evaluated at the possible resource shares corresponding to the different subsets that he is in, the entire function need not be elicited. The mechanism only needs to solicit the utilities of the buyers for the different shares he may receive. This corresponds to an elicitation of $2^{n-1}+1$ values in the worst case for each buyer, but can be much lower depending on the choice of these shares. In the examples that we considered, the resource share of each buyer corresponding to a set $A$ is $\frac{1}{|A|}$ if this buyer is in $A$. For this choice of shares, only $n$ utility values corresponding to the shares $\{\frac{1}{i}; i=1,\cdots,n\}$ are needed to be reported. The mechanism nevertheless needs to ensure that the values for the different shares are samples of some function in $\mathcal{C}$. In order to ensure this, it is sufficient to check that the piecewise linear extrapolation of the reported values that results in a function on $[0,1]$ is in $\mathcal{C}$, i.e. it is concave and non-decreasing. 

Following is our main result.

\begin{theorem}\label{thm1}
Suppose that the resource is being sold in a sealed-bid second price auction with a reserve price. Also, assume that a buyer strictly prefers the outcome where he obtains a non-zero fraction of the resource with a payment equal to his utility for that fraction of the resource, to the outcome where he does not obtain anything and makes no payment. Then any monotonic aggregation mechanism is coalition-strategyproof.
\end{theorem}

The mechanism thus truthfully elicits individual preferences and then aggregates them to result in a single joint decision-making entity. In order to prove this theorem, we formalize a few ideas.

In a second price auction with a reserve price, there is a minimum price needed to be paid by an agent to win the resource, which is the maximum of the reserve price and the highest bid of all the other agents. For our group of buyers, let this price be $p^*$. For two different sets of reports of the utility functions by the buyers in the group, we say that outcome of the auction remains the same if the same set of buyers $W\subseteq L$ is the winning set. Since the resource and payment shares depend only on the set of winning buyers, these shares in the two outcomes are the same for all the buyers. Consider a coalition of buyers $C\subseteq L$. Assume that the reports of the utility functions of all the other buyers are fixed. Suppose that if all the agents in $C$ report their utilities truthfully (keeping all other reports fixed), then the vector of values generated is $\bar{\beta}^0=(\beta^0_1,\beta^0_2,\cdots,\beta^0_m)$ and let the winning set of buyers be $W^0$ (which may be empty). Also let $\overline{C}$ be the set of agents in $C$ who are in $W^0$ and let $C'=C\setminus \overline{C}$. 
We first prove the following lemma.

\begin{lemma}\label{int1}\begin{enumerate}
\item 
Suppose that starting from a set of buyers $S$, under a fixed report from all the buyers, the set of winning buyers is some set $W\subseteq S$ (which may be empty). Then starting from a set of buyers $S\setminus M$ where $M \subseteq S \setminus W$, under the same reports, the set of winning buyers is also $W$. 
\item
Suppose that starting from a set of buyers $S$, under a fixed report from all the buyers, the set of winning buyers is some set $W\subseteq S$. Then starting from a set of buyers $S\setminus\{i\}$ where $i\in W$, the set of winning buyers $W'$ satisfies $W'\subseteq G\setminus \{i\}$.
\end{enumerate}
\end{lemma}
\begin{proof}
We first prove the first claim. First, we can easily show that $W$ is a subset of the new set of winning buyers $W'$. This is because, if any buyer $i$ in $W$ claims to be able to afford to pay $y_i(W)p^*$ for $x_i(W)$ fraction of the resource, he can also pay $y_i(W\cup Q)p^*$ for $x_i(W\cup Q)$ fraction of the resource. This is because, by the converse implication of the monotonicity assumption, we have that 
if $G_i(x_i(W))\geq Cy_i(W)$ then $G_i(x_i(W\cup Q)) \geq Cy_i(W\cup Q)$ for every $i$ and every $G_i \in \mathcal{C}$.

Thus we just need to prove that $Q=W'\setminus W$ is empty. Suppose not. Then there is some buyer $i\in Q$ who in the original case was removed from the set of buyers $S'\subseteq S$ where $(W\cup Q)\subseteq S'$. This in particular implies that $G_i(x_i(S'))<p^*y_i(S')$. But this again implies, by the monotonicity assumption that $G_i(x_i(W\cup Q))<p^*y_i(W\cup Q)$. This contradicts the assumption that $i$ is in the new winning set. Thus $Q$ has to be empty. For the second claim, we just need to prove that $Q=W'\setminus (W\setminus \{i\})$ is empty, which agains follows from a similar argument.\end{proof}

Now suppose for some fixed reports of all buyers, the vector of values computed by the mechanism is $\bar{\beta}$. Now the the effect of any deviation from these fixed reports by a coalition $C$, manifests itself for the first time by a change in some $\beta_k$ corresponding to some subset $S_k$, to a new value $\beta'_k$. If $\beta'_k<\beta_k$ then this change has been effectively implemented by one buyer $i$ in coalition $C$ by under-reporting. If $\beta'_k>\beta_k$, then there must be at least one buyer $i$ in the coalition $C$ who was forcing the constraint in the set $S_k$ under truthful reporting, i.e. $G_i(x_i(S_k))=\beta_ky_i(S_k)$ and who over-reports. In either case, we say that such a buyer in $C$ is responsible for causing this first change. Knowing this, we decompose the deviation by coalition $C$ into a sequence of deviations by individual agents that sequentially bring out the transformation:
$$\bar{\beta}^0\rightarrow \bar{\beta}^1\rightarrow \bar{\beta}^2\rightarrow\cdots \bar{\beta}^N.$$ 
This sequence is constructed in the following way. First, $\bar{\beta}^0$ is the vector of values computed under truthful reports from all agents in $C$ and certain assumed fixed reports of the other buyers. The individual agent in $C$ whose non-truthful report brings out the first change in $\bar{\beta}^0$ is denoted by $i_1$. Next $\bar{\beta}^1$ is the vector of values computed under truthful reports by all the buyers in $C\setminus i_1$, the non-truthful report of agent $i_1$ and under the assumed fixed reports of the other buyers. Then recursively, we denote $i_j$ to be the individual agent in $C$ whose non-truthful report brings out the first change in $\bar{\beta}^{j-1}$. Then define $\bar{\beta}^j$ to be the vector of values computed under truthful reports by all the buyers in $C\setminus \{i_0,i_1,\cdots, i_j\}$, the non-truthful report of agents $\{i_0,i_1,\cdots, i_j\}$ and under the assumed fixed reports of the other buyers. Finally $\bar{\beta}^n$ is the vector of values computed using the deviated reports of all the agents in $C$ (note that not all agents in $C$ are necessarily non-truthful) and the assumed fixed reports of the other buyers. With each $\bar{\beta}^{j}=(\beta^j_1,\cdots,\beta^j_{m_j})$, are the associated subsets $\{S^j_1,\cdots,S^j_{m_j}\}$ encountered by the mechanism. Further denote $W^j=S^j_{r_j}$ to be the winning set of buyers corresponding to $\bar{\beta}^{j}$, where $r_j=\min\{i: \beta^j_i\geq p^*\}$. We first prove the following.
\begin{lemma}
For $j\in\{1,\cdots,n\}$, if  $i_j\notin W^{j-1}$ then either $W^j=W^{j-1}$ or $W^j=W^{j-1}\cup M \cup \{i_j\}$, where $M\cap W^{j-1}=\phi$ and $i_j$ makes a strict loss being in $W^j$.
\end{lemma}
\begin{proof}
Let $k$ be the smallest index at which $\bar{\beta}^{j}$ differs from $\bar{\beta}^{j-1}$. Note that $S^{j-1}_k=S^{j}_k$. Now if agent $i_j$ changes the value of $\beta^{j-1}_k$ to $\beta^j_k<\beta^{j-1}_k$, then he forces the constraint to have $G_{i_j}(x_{i_j}(S^{j}_k))=\beta^1_ky_{i_j}(S^{j}_k)$. Since $i_j$ is not in the winning set of buyers under $\bar{\beta}^{j-1}$, $\beta^{j-1}_k<p^*$ and thus $\beta^{j}_k<p^*$. Thus agent $i_j$ is removed from the group and the remaining set of buyers is $S^{j}_k\setminus i_j$. By the first claim in lemma \ref{int1}, the set of winning buyers is again $W^{j-1}$. Next, assume that the agent $i_j$'s report changes the value of $\beta^{j-1}_k$ to $\beta^j_k>\beta^{j-1}_k$. This implies that $i_j$ was forcing the constraint under truthful reporting in step $j-1$, i.e. $U_{i_j}(x_{i_j}(S^{j-1}_k))=\beta^{j-1}_ky_{i_j}(S^{j-1}_k)<p^*y_{i_j}(S^{j-1}_k)$. Now if $\beta^j_k\geq p^*$, then the resource is allotted to the group $S^{j-1}_k=S^{j}_k$ and thus agent $i_j$ is a part of the new group of winning buyers $W^j$, resulting in a strict loss $p^*y_{i_j}(S^{j}_k)-U_{i_j}(x_{i_j}(S^{j}_k)>0$ for him. Thus $W^{j}=W^{j-1}\cup M\cup i_j$ where $M\cap W^{j-1}=\phi$ and $i_j$ makes a strict loss being in $W^{j}$. 

Next if $\beta^j_k< p^*$, then in the case that $i_j$ is forcing the constraint by his report, i.e. $G_{i_j}(x_{i_j}(S^{j}_k))=\beta^j_ky_{i_j}(S^{j}_k)$, he is still removed from the set of buyers and by the first claim in lemma \ref{int1}, the set of winning buyers is again $W^{j-1}$. In the case that $i_j$ does not force the constraint by his report, i.e. $G_{i_j}(x_{i_j}(S^{j}_k))>\beta^{j}_ky_{i_j}(S^{j}_k)$, it follows that some other buyer $m$ was removed from the set of buyers. It cannot be one of the buyers in $W^{j-1}$, since if $y_m(W^{j-1})p^*\leq G_m(x_m(W^{j-1}))$, by the converse implication of monotonicity, $y_m(S^{j}_k)p^*\leq G_m(x_m(S^{j}_k))$ also and thus, since $\beta^j_k< p^*$, we have that $y_m(S^{j}_k)\beta^j_k\leq G_m(x_m(S^{j}_k))$. Thus the agent $m$ is in $S^{j}_k\setminus W^{j-1}\cup i_j$. Now with the new $\bar{\beta}^j$ computed, either $i_j$ is in the winning set of buyers, in which case he makes a strict loss since if $p^*y_{i_j}(S^{j}_k)>U_{i_j}(x_{i_j}(S^{j}_k))$, then he cannot afford to pay $p^*y_{i_j}(S^{j}_r)>U_{i_j}(x_{i_j}(S^{j}_r))$, for any $S^{j}_r\subset S^{j}_k$. Also $W^{j-1}$ has to be in the winning set again by the converse implication of the monotonicity assumption. Thus the claim holds true. Or he is not in the winning set of buyers. In that case suppose he was removed from some set $S^{j}_q$ at some stage $q$ by the mechanism. Then note that since $S^{j}_q$ is not the winning set, the entire set $W^{j-1}$ is has to be in the set of active buyers $S^{j}_{q+1}$, again because of the converse implication of the monotonicity assumption. Thus after stage $q$, the set of remaining buyers is $S^{j}_{q+1}=W^{j-1}\cup Q$ for some $Q\subseteq S^{j}_k\setminus W^{j-1}$ (which may be empty). Thus by the first claim in lemma \ref{int1}, the the set of winning buyers is again $W^{j-1}$. Thus again the claim holds true.
\end{proof}

\begin{lemma}\label{main}
For each $j$ at least one of the following holds.
\begin{enumerate}
\item $W^j=W^0$.
\item $(W^0\setminus W^j)\cap \overline{C}\neq \phi$. 
\item $W^{j} \cap C \neq \phi$ and at least one $i\in W^j\cap C$ makes a strict loss being in $W^{j}$.
\end{enumerate}
 \end{lemma}
 
\begin{proof}
We will prove this lemma using induction. The claim clearly holds for $j=0$. Now we assume that at least one of the following hypotheses hold true for step $j-1$ for some $j\geq 2$:
\begin{enumerate}
\item $W^{j-1}=W^0$.
\item $(W^0\setminus W^{j-1})\cap \overline{C}\neq \phi$. 
\item $W^{j-1} \cap C \neq \phi$ and at least one $i\in W^j\cap C$ makes a strict loss being in $W^{j}$.
\end{enumerate}

If the agent $i_j\notin W^{j-1}$, then the previous lemma says that either $W^j=W^{j-1}$ or $W^j=W^{j-1}\cup M' \cup \{i_j\}$ where $M' \cap W^{j-1}=\phi$ and $i_j$ makes a strict loss. Since $i_j\in C$, in any case at least one of the hypotheses holds true for step $j$.

Next, suppose that the agent $i_j$ is in $W^{j-1}$. Suppose that the first change in the vector $\bar{\beta}^{j-1}$ happens at index $k$ and the changed value $\beta^j_k$ is such that $\beta^j_k>\beta^{j-1}_k$.  This means that $i_j$ was forcing the constraint under truthful reporting in step $j-1$, and since he is amongst the winning set of buyers $W^{j-1}$, this means that $S^{j-1}_k$ is that winning set and still remains so under $\bar\beta^j$. Hence $W^j=W^{j-1}$ and thus one of the hypotheses is true for step $j$. 

Now suppose that $\beta^j_k<\beta^{j-1}_k$.  Then either $\beta^{j}_k\geq p^*$, in which case $\beta^{j-1}_k\geq p^*$ also and $S^{j-1}_k$ was and still remains the winning set of buyers, which again implies that one of the hypotheses is true for step $j$. Or $\beta^j_k< p^*$, in which case agent $i_j$ is removed from the subset. Now from the second claim of lemma \ref{int1}, the set of winning buyers satisfies $W^j\subseteq W^{j-1}\setminus\{i_j\}$. In the case that $W^{j-1}=W^0$ note that $i_j\in \overline{C}$ and thus $W^{j}\subseteq W^0\setminus\{i_j\}$ implies that the second hypothesis is true. In the case that $(W^0\setminus W^{j-1})\cap \overline{C}\neq \phi$, again the second hypothesis is true. 

In the case that $W^{j-1} \cap C \neq \phi$ and at least one $i\in W^j\cap C$ makes a strict loss being in $W^{j}$, suppose that $W^{j-1}=A\cup M$ where $A\subseteq W^0$ and $M\cap W^0=\phi$. Then we have that $W^j=A'\cup M'$ where $A'\subseteq A \setminus\{i_j\}$ and $M'\subseteq M\setminus\{i_j\}$ and $M'\cap W^{0}=\phi$. Now if $i_j\in A$ then the second hypothesis is true. Suppose then that  $i_j\in M$. 

Now each buyer $i$ in $M'$ is one of two types:
\begin{enumerate}
\item Type $A$:  $i\in L\setminus C$.
\item Type $B$: $i\in \{i_{1},\cdots,i_n\}\setminus \{i_j\}$.
\end{enumerate}

Now since $M'$ is disjoint from $W^{0}$, there must be a buyer $m\in M'$ who at step $0$ was removed from some set of buyers $S'$ such that $W^{0}\cup M'\subseteq S'$. This in particular implies that $G_m(x_m(S'))<p^*y_m(S')$. But this again implies, by the monotonicity assumption that $G_m(x_m(A'\cup M'))<p^*y_m(A'\cup M')$. Now if this buyer $m$ is of type $A$, then his report at step $j$ is the same as his report at step $0$, which is $G_m$. And thus this contradicts the assumption that $m$ is in the new winning set $W^j$. Thus the buyer $m\in M$ who at step $0$ was removed from the set of buyers $S'$, such that $W^{0}\cup M\subseteq S'$, has to be of type $B$. But such a buyer $m$'s report at step $0$ is truthful. Thus this implies that $U_m(x_m(S'))<p^*y_m(S')$. Thus by the monotonicity assumption, $U_m(x_m(A'\cup M'))<p^*y_m(A'\cup M')$ and thus he makes a strict loss being in $W^j$. Thus either there is at least one buyer in $M'$ who faces a strict loss being in $W^j$ and he is in $C$ (since a buyer of type $B$ is in $C$) or $M'=\phi$. In the prior case at least one of the hypotheses holds for step $j$. 

Consider the latter case where $M'=\phi$. Suppose that $A'\subseteq A\subset W^0$. Then there must be a buyer $m\in W^0\setminus A'$ who at step $j$, was removed from the set of buyers $\tilde{S}$ such that $W^0 \subseteq \tilde{S}$. This in particular implies that $G_m(x_m(S'))<p^*y_m(S')$ where $G_m$ is this buyer's report at step $j$. But this again implies, by the monotonicity assumption that $G_m(x_m(W^0))<p^*y_m(W^0)$. Now if this buyer $m$ is of type $A$, then his report at step $0$ is the same as his report at step $j$, which is $G_m$. And thus this contradicts the assumption that $m$ was in the winning set $W^0$. Hence $m$ must be of type $B$ and since $m\in W^0$, he also must be in $\overline{C}$. And thus either $(W^0\setminus A' \cap \overline{C})\neq \phi$ or $A'=A=W^0$. Hence again one of the hypotheses holds true.
\end{proof}

{\bf Proof of theorem \ref{thm1}:} Lemma \ref{main} shows that that one of the following holds true:
\begin{enumerate}
\item $W^{N}=W^0$.
\item $(W^0\setminus W^{N})\cap \overline{C}\neq \phi$. 
\item $W^{N} \cap C \neq \phi$ and at least one $i\in W^N\cap C$ makes a strict loss being in $W^{N}$.
\end{enumerate}In the first case, the deviation by the coalition $C$ has no effect on the outcome of the auction. In the second case, there is at least one buyer in $C$ that was in the winning set under truthful reporting and is not because of the deviation. Since any buyer strictly prefers winning, even if it is under a zero net profit, over not winning anything, such a deviation is not supported by such a buyer. In the third case, there is a buyer in $C$ that makes a strict loss by deviation and hence again this buyer would not support such a deviation. Hence the mechanism is coalition-strategyproof. \qed

There is an implicit assumption in our model that the buyers are individually `small', in the sense that they are not expected to be able to pay for the entire resource on their own, i.e. $U_i(1)$ is small as compared to the typical highest bid of the other agents (or the reserve price). This justifies the assumption that if a buyer is excluded by the mechanism from participating in sharing the resource, he cannot compete against the group on his own for the entire resource, or that he is willing to participate in the mechanism since his `outside option' has no value. But nevertheless, it is interesting to see what the mechanism has to offer to a buyer who could have potentially purchased the entire resource for himself. We define the following notion of `individual consistency' of an aggregation mechanism.

\begin{definition} ({\bf Individual consistency}) Consider the set of buyers $W$ that would have individually been willing to purchase the entire resource for the price $p^*$ offered in the external auction, i.e. for all $i\in W$, $U_i(1)>p^*$.  Then the aggregation mechanism is said to be \emph{individually consistent} if $W\neq \phi$ implies that the group ends up purchasing the resource.
\end{definition}
We then have the following result.
\begin{theorem}\label{thm1}
Any monotonic aggregation mechanism with fixed shares is individually consistent. Further if $W$ is the set of buyers for whom $U_i(1)>p^*$, then all the buyers in $W$ are in the winning set of buyers $S_r$.
\end{theorem}
\begin{proof}
Suppose there is a buyer $i$ for whom $U_i(1)\geq p^*$. From the definition of the resource and payment shares, $x_i(\{i\})=y_i(\{i\})=1$ and thus  this inequality is the same as $U_i(x_i(\{i\}))\geq y_i(\{i\})p^*$. But the monotonicity condition conversely implies that if 
$$U_i(x_i(\{i\}))\geq y_i(\{i\})p^*,$$
then 
\begin{equation}\label{lolo}
U_i(x_i(B))\geq y_i(B)p^*
\end{equation}
for any $B$ such that $i\in B$. Now $i\in S_1=L$ and further (\ref{lolo}) implies that $i\notin T_j$ for any $S_j$ such that the corresponding $\beta_j<p^*$. Further $U_i(x_i(\{i\}))\geq y_i(\{i\})p^*$ implies that there is some $j^*$ such that $\beta_{j^*}\geq p^*$ and $\beta_{j}<p^*$ for all $j<j^*$ and further $\{i\}\in S_{j^*}$. Thus the resource is bought in the outcome of the mechanism. This argument is valid for every buyer $i\in W$. Thus each buyer in $W$ is in the winning set of buyers.
\end{proof}
Note that the mechanism guarantees that the entire set $W$ is in the winning set of buyers $S_r$. But the share of the resource for a buyer in $W$ in the winning set could be $0$ (and hence the payment is also $0$). This can be easily avoided by choosing resource shares that satisfy $x_i(A)>0$ for all $i\in A$, for every $A\in L$. 

\section{Explicit Characterizations}
In this section we give explicit sharing schedules that satisfy the monotonicity condition with respect to class $\mathcal{C}$. For the first kind consider the following schedule of shares: \\
\hrule
\noindent\textbf{Cross-monotonic sharing schedule for the class of utility functions $\mathcal{C}$ (CMSS)}\\ For each subset $A\subseteq L$, fix $n$ non-negative numbers $(x_1(A),\cdots,x_n(A))$ which are the resource as well as the payment shares, such that they satisfy 
\begin{enumerate}
\item $\sum_{i=1}^n x_i(A)=1$ and $x_i(A)>0$ only if $i\in A$. 
\item (Cross-monotonicity) If $A\subseteq B$, then $x_i(A)\geq x_i(B)$ for all $i\in A$.
\end{enumerate}
\hrule\bigskip
Thus in the cross-monotonic sharing schedule, the payment shares and the resource shares are equal for each subset of buyers, and these shares satisfy the property of cross-monotonicity. We can then show that this choice of shares also satisfies the monotonicity requirement with respect to the class $\mathcal{C}$.

\begin{theorem} The choice of resource and payment shares in the cross-monotonic sharing schedule satisfies the monotonicity condition with respect to the class $\mathcal{C}$.
\end{theorem}
\begin{proof}
By the property of any concave function $U\in \mathcal{C}$, if $U(x)< Cx$ for some $C>0$ and some $x\in[0,1]$, then $U(x')< Cx'$ for any $x\geq x'$ (see figure). Then the result follows from the cross-monotonicity of the shares.
\end{proof}
\begin{figure}[htb]
\begin{center}
\includegraphics[width=3in,angle=0]{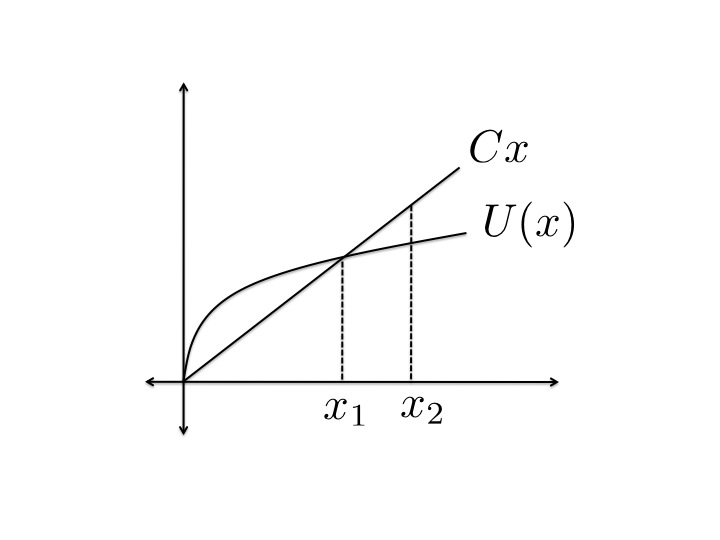}
\caption{The single crossing property of linear functions with respect to concave functions: if $U(x_1)< Cx_1$ for some $C>0$ and some $x_1\in[0,1]$, then $U(x_2)< Cx_2$ for any $x_2\geq x_1$}
\end{center}
\end{figure}

In order to define the next class of mechanisms, we define the following property. 
\begin{definition} (Single crossing property) A function $f:[0,1]\rightarrow \mathbb{R}_+$ is said to satisfy the single crossing property with respect to a class of functions $\mathcal{K}$ if for every $C > 0$ and every $U\in \mathcal{K}$, if $Cf(x)> U(x)$ for some $x\in [0,1]$, then $Cf(x')> U(x')$ for any $x'>x$.
\end{definition}
Note that a function $f(x)$ satisfies the single crossing property with a class of functions $\mathcal{K}$, if and only if any function $Cf(x)$ for $C> 0$ also satisfies this property.  We saw that the primary reason why equal payment and resource shares satisfy the monotonicity condition when these shares are cross-monotonic, is that the function $f(x)=x$ satisfies the single crossing property with respect to the class of utility functions $\mathcal{C}$. 
Utilizing this property, we can provide another explicit characterization of the monotonic aggregation mechanism. \\
\hrule
\noindent {\bf Ranked resource allocation schedule (RRAS)}\\
Suppose that $\tilde{\mathcal{C}}$ is a subset of the class $\overline{\mathcal{C}}$ and a concave function $f:[0,1]\rightarrow R_+$ satisfies the single crossing property with respect to this class. Then consider the following choice of shares in a monotonic aggregation mechanism:
\begin{itemize}
\item {\bf Resource shares:} 
\begin{itemize}
\item  Fix a ranked ordering of the buyers $\{1,\cdots,n\}$. Fix an initial set of resource shares $(x_1(L),\cdots,x_n(L))$ corresponding to set $L$ satisfying\\ $\sum_{i=1}^n x_i(L)=1$.  
\item For each subset $A\subset L$, the rank ordered vector of resource shares of buyers in $A$ is defined to be
$$(x_{j_1}(A),\cdots, x_{j_l}(A))=((1-\sum_{r=2}^l x_{j_r}(L)), x_{j_2}(L),\cdots,x_{j_l}(L)).$$ 
\end{itemize}
\item {\bf Payment shares:}
\begin{itemize}
 \item For the payment shares, for each subset $A\subset L$ and buyer $i$, define $$y_i(A)=\frac{f(x_i(A))}{\sum_{j\in A}f(x_j(A))}$$
\end{itemize}
\end{itemize}
\hrule\bigskip
The resource-sharing scheme of this schedule has the following interpretation. When a set of buyers diminishes to a smaller subset in the monotonic aggregation mechanism, the shares of the buyers that are removed are allocated to the buyer with the highest rank in the subset. Once the shares of the buyers in the largest set $L$ is fixed, this rule determines the shares corresponding to all of its subsets. The payment shares in any subset are defined to be proportional to the values of the function $f$ evaluated at the corresponding resource shares. We can again show that under the assumptions on the properties satisfied by $f$, this sharing schedule satisfies the monotonicity condition.
\begin{theorem}
The choice of resource and payment shares in the ranked resource allocation schedule satisfies the monotonicity condition with respect to the class $\tilde{\mathcal{C}}$.
\end{theorem}
\begin{proof}
We need to show that for any $C>0$, for any two subsets $A$ and $B$ such that $A\subseteq B$, and for any $i\in A$, if $U_i(x_i(B))< C\frac{f(x_i(B))}{\sum_{i\in B}f(x_j(B))}$ then $U_i(x_i(A))< C\frac{f(x_i(A))}{\sum_{j\in A}f(x_j(A))}$ for every $U_i \in \mathcal{C}$. Now, the choice of the resource shares is such that $x_i(A)\geq x_i(B)$ and so from the single crossing property of the function $f$, we have that 
$$U_i(x_i(A))< C\frac{f(x_i(A))}{\sum_{j\in B}f(x_j(B))}\textrm{ for every }U_i \in \tilde{\mathcal{C}}.$$
Thus we need to show that $\sum_{j\in B}f(x_j(B))\geq \sum_{j\in A}f(x_j(A))$. Let $i^*$ be the highest ranked buyer in the set $A$ and let $(x_{j_1}(b),\cdots,x_{j_{|B\setminus A|}})$ be the ordered vector of shares of buyers in $B\setminus A$. Then from the definition of the resource shares, since the shares of the buyers in $A\setminus{i^*}$ are the same in the sets $A$ and $B$ and since in $A$, $i^*$ gets the all the shares of the buyers in $B\setminus A$, we have that
$$\sum_{j\in B}f(x_j(B))- \sum_{j\in A}f(x_j(A))=\sum_{r=1}^{|B\setminus A|}f(x_{j_r}) + f(x_{i^*})- f(\sum_{r=1}^{|B\setminus A|}x_{j_r}+x_{i^*})$$
\begin{eqnarray*}
&=&\sum_{r=1}^{|B\setminus A|}\bigg(f(x_{j_r})-f(0)\bigg)-\bigg(f(\sum_{l=1}^{r-1}x_{j_l}+x_{i^*} + x_{j_r})-f(\sum_{l=1}^{r-1}x_{j_l}+x_{i^*})\bigg)+f(0)\\
&\geq& \sum_{r=1}^{|B\setminus A|} f(0)\geq 0.
\end{eqnarray*}
The first inequality holds since $f$ is concave and the shares are non-negative, and the second holds since $f$ is non-negative.
\end{proof}
Note that the resource shares in the RRAS mechanism are cross-monotonic. Thus we can naturally compare this mechanism to the CMSS with the same resource shares.
\begin{example}
Consider the class of utility functions $$\tilde{\mathcal{C}}=\{g(x)=cx^k: c\geq 0,\, k\in[0,\frac{1}{2}]\}.$$ One can easily show that the function $f(x)=\sqrt{x}$ satisfies the single crossing property with respect to this class. Consider three buyers 1, 2 and 3 with utility functions $U_1(x)=x^{\frac{1}{4}}$, $U_2(x)=x^{\frac{1}{3}}$ and $U_3(x)=\sqrt{x}$ respectively. Suppose that the priority order of the buyers is $\{1,2,3\}$. Let the resource shares corresponding to the entire set of the buyers be $(\frac{1}{2},\frac{1}{4},\frac{1}{4})$. This determines the resource and payment shares for all the subsets, as described in the priority based aggregation mechanism. The following table shows these shares for the three buyers. The last two columns show the sequence of $\{\beta_j\}$ computed using the RRAS shares and the CMSS shares (in which the payment shares are the same as the resource shares) respectively. The $\{\beta_j\}$ are only computed for the sets $\{S_j\}$ that are encountered in the mechanism. Thus as shown in the table, for the RRAS shares, $S_1=\{1,2,3\}$, $S_2=\{1,2\}$ and $S_3=\{2\}$ while for the CMSS shares, $S_1=\{1,2,3\}$, $S_2=\{2,3\}$ and $S_3=\{3\}$.
\begin{table}[ht]
\centering
\begin{tabular}{ccccc}\hline\hline
Subsets     & Resource shares          & RRAS Payment shares    & RRAS $\{\beta_j\}$ & CMSS $\{\beta_j\}$ \\
\hline

$\{1,2,3\}$ & $(\frac{1}{2},\frac{1}{4},\frac{1}{4})$ & $(\frac{1}{1+\sqrt{2}},\frac{1}{2+\sqrt{2}},\frac{1}{2+\sqrt{2}})$ & 1.707              & 1.68               \\
$\{1,2\}$    & $(\frac{3}{4},\frac{1}{4},0)$  & $(\frac{\sqrt{3}}{1+\sqrt{3}},\frac{1}{1+\sqrt{3}},0)$  & 1.467  & -                  \\
$\{2,3\}$   &$(0,\frac{3}{4},\frac{1}{4})$ & $(0,\frac{\sqrt{3}}{1+\sqrt{3}},\frac{1}{1+\sqrt{3}})$ & -               & 1.21               \\
$\{1,3\}$   & $(\frac{3}{4},0,\frac{1}{4})$ & $(\frac{\sqrt{3}}{1+\sqrt{3}},0,\frac{1}{1+\sqrt{3}})$ & -   & -                  \\
$\{1\}$     & $(1,0,0)$             & $(1,0,0)$                      & -                  & -                  \\
$\{2\}$     & $(0,1,0)$            & $(0,1,0)$                      & 1                  & -                  \\
$\{3\}$     & $(0,0,1)$            & $(0,0,1)$                      & -                  &  1 \\
\hline                 
\end{tabular}
\end{table}
Note that the ordered vector of $\beta$ values computed in RRAS mechanism dominates the vector of values computed in the CMSS mechanism .
\end{example}


\section{Conclusion }
We designed a class of preference aggregation mechanisms that we call monotonic aggregation mechanisms that enable a group of agents with private utilities to make purchasing decisions for a shared resource. The key properties of the mechanism are coalition-strategyproofness, individual consistency and the ability to exactly recover the price of the resource in the market. We gave two explicit characterizations of this mechanism for the case where the utility functions of the buyers are concave and non-decreasing. Apart from these characterizations, one may find other explicit characterizations of the monotonic aggregation mechanism for specific classes of utility functions. Generalized techniques of designing resource and payment sharing schemes that satisfy the monotonicity condition for broad classes of utility functions would be very useful. 
The monotonic aggregation mechanism is not efficient. It may be the case that the group of agents may jointly be able to pay for the resource with some payment and resource division, but the mechanism leads to an outcome in which they do not buy the resource. Although, this efficiency loss is unavoidable in mechanisms that satisfy participation constraints and require an exact recovery of the fixed cost, an insight into the worst-case welfare properties of our mechanism would be very helpful.
\bibliographystyle{plain}

\bibliography{mybibfile}

\begin{thebibliography}{10}

\bibitem{incentive_auction_ro}
In the {M}atter of {E}xpanding the {E}conomic and {I}nnovation {O}pportunities
  of {S}pectrum {T}hrough {I}ncentive {A}uctions, June 2014.

\bibitem{aumann1974values}
Robert~J Aumann and Lloyd~S Shapley.
\newblock {\em Values of non-atomic games}, volume 189.
\newblock Princeton University Press Princeton, 1974.

\bibitem{billera1982allocation}
Louis~J Billera and David~C Heath.
\newblock Allocation of shared costs: a set of axioms yielding a unique
  procedure.
\newblock {\em Mathematics of Operations Research}, 7(1):32--39, 1982.

\bibitem{bykowsky2008market}
Mark Bykowsky, Mark Olson, and William Sharkey.
\newblock A market-based approach to establishing licensing rules: Licensed
  versus unlicensed use of spectrum.
\newblock {\em Office of Strategic Planning and Policy Analysis Working Paper},
  43, 2008.

\bibitem{Bykowsky10}
Mark Bykowsky, Mark Olson, and William Sharkey.
\newblock Efficiency gains from using a market approach to spectrum management.
\newblock {\em Information Economics and Policy}, 22(1):73--90, 2010.

\bibitem{friedman1999three}
Eric Friedman and Herve Moulin.
\newblock Three methods to share joint costs or surplus.
\newblock {\em Journal of Economic Theory}, 87(2):275--312, 1999.

\bibitem{koster2007moulin}
Maurice Koster.
\newblock The moulin--shenker rule.
\newblock {\em Social Choice and Welfare}, 29(2):271--293, 2007.

\bibitem{milgrom2011case}
Paul Milgrom, Jonathan Levin, and Assaf Eilat.
\newblock The case for unlicensed spectrum.
\newblock {\em Policy Analysis}, 2011.

\bibitem{mirman1982demand}
Leonard~J Mirman and Yair Tauman.
\newblock Demand compatible equitable cost sharing prices.
\newblock {\em Mathematics of Operations Research}, 7(1):40--56, 1982.

\bibitem{moulinshenker99}
Herv$\acute{e}$ Moulin and Scott Shenker.
\newblock Strategyproof sharing of submodular costs: budget balance versus
  efficiency.
\newblock {\em Economic Theory}, 18(3):511--533, 2001.

\bibitem{moulin1995additive}
Herve Moulin.
\newblock On additive methods to share joint costs.
\newblock {\em Japanese Economic Review}, 46(4):303--332, 1995.

\bibitem{moulin1999incremental}
Herv{\'e} Moulin.
\newblock Incremental cost sharing: Characterization by coalition
  strategy-proofness.
\newblock {\em Social Choice and Welfare}, 16(2):279--320, 1999.

\bibitem{moulin1992serial}
Herve Moulin and Scott Shenker.
\newblock Serial cost sharing.
\newblock {\em Econometrica: Journal of the Econometric Society}, pages
  1009--1037, 1992.

\bibitem{moulin2007fair}
Herv{\'e} Moulin and Yves Sprumont.
\newblock Fair allocation of production externalities: recent results.
\newblock {\em Revue d'{\'e}conomie politique}, 117(1):7--36, 2007.

\bibitem{shapley1952value}
Lloyd~S Shapley.
\newblock A value for n-person games.
\newblock Technical report, DTIC Document, 1952.

\bibitem{shubik1962incentives}
Martin Shubik.
\newblock Incentives, decentralized control, the assignment of joint costs and
  internal pricing.
\newblock {\em Management science}, 8(3):325--343, 1962.

\bibitem{sprumont1998ordinal}
Yves Sprumont.
\newblock Ordinal cost sharing.
\newblock {\em Journal of Economic Theory}, 81(1):126--162, 1998.

\bibitem{sprumont2010axiomatization}
Yves Sprumont.
\newblock An axiomatization of the serial cost-sharing method.
\newblock {\em Econometrica}, 78(5):1711--1748, 2010.

\end{thebibliography}

\end{document}